\documentclass{svproc}
\usepackage{url}
\usepackage[T1]{fontenc}
\usepackage{libertine}
\usepackage[scaled=0.85]{beramono}
\usepackage{amsmath}
\usepackage{amssymb}
\usepackage{graphicx}
\usepackage{mathrsfs}

\makeatletter
\newsavebox{\@brx}
\newcommand{\llangle}[1][]{\savebox{\@brx}{\(\m@th{#1\langle}\)}%
  \mathopen{\copy\@brx\kern-0.5\wd\@brx\usebox{\@brx}}}
\newcommand{\rrangle}[1][]{\savebox{\@brx}{\(\m@th{#1\rangle}\)}%
  \mathclose{\copy\@brx\kern-0.5\wd\@brx\usebox{\@brx}}}
\makeatother

\begin{document}
\mainmatter             

\title{LR Parsing of Permutation Phrases}

\titlerunning{LR Parsing of Permutation Phrases} 
\author{Jana Kostičová}
\authorrunning{Jana Kostičová} 

\institute{Comenius University, Bratislava, Slovakia,\\
\email{kosticova@dcs.fmph.uniba.sk},\\ WWW home page:
\texttt{http://www.dcs.fmph.uniba.sk/$\sim$kosticova/}
}

\maketitle              
																																			\begin{abstract}
This paper presents an efficient method for LR parsing of permutation phrases. In practical cases, the proposed algorithm constructs an LR(0) automaton that requires significantly fewer states to process a permutation phrase compared to the standard construction. For most real-world grammars, the number of states is typically reduced from $\Omega(n!)$ to $O(2^{n})$, resulting in a much more compact parsing table. The state reduction increases with longer permutation phrases and a higher number of permutation phrases within the right-hand side of a rule. We demonstrate the effectiveness of this method through its application to parsing a JSON document.
\keywords{permutation phrase, LR parsing, state complexity, unordered content, JSON}
\end{abstract}

\section{Introductions}
Several of today's languages allow for constructs consisting of unordered content, meaning that any permutation of child-subconstructs is allowed. Examples of such languages include Java, Haskell, XML, 
and JSON (JavaScript Object Notation) \cite{ecma:json,ietf:json}. JSON, an extremely popular tree format for data storage and transmission, is perhaps the most prominent example, as JSON objects always consist of zero or more unordered members.  The structure of JSON data can be constrained by a schema, most commonly using JSON Schema \cite{jsonschema}. A large part of such schema can be expressed using a context-free grammar (CFG) \cite{hopcroft2013:automata}.

EBNF notation for CFGs requires all permutation options to be enumerated on the right-hand side of the rule that results in $n!$ grammar rules. For example, an unordered content over the symbols $A, B, C$ yields $3! = 6$ rules: $S \to ABC | ACB | BAC |BCA$ $| CAB | CBA$.  Cameron, in his work \cite{cameron1993:permphrases}, proposed a shorthand notation for expressing unordered content, called a permutation phrase: 
$X \to \llangle\, A \,\|\, B \,\|\, C  \,\rrangle$.
Using this notation does not affect the  expressiveness of CFGs, nor does it change the language generated by the grammar. However, it makes the language specification significantly more concise and easier to understand. Consequently, it is desirable to adapt common parsing algorithms to accept permutation phrases in the input grammar and to process them more efficiently than the original algorithms, which handled $n!$ rules.

The main contribution of this paper is a modification to LR parsing algorithms that enables efficient parsing of permutation phrases.  The running time remains unaffected, and in practical cases, the number of states in the LR(0)/LR(1) automaton, as well as the size of the resulting parsing table, are significantly reduced compared to the original algorithm. This state reduction is achieved by changing the semantics of LR(0) items: instead of tracking the exact sequence of symbols already seen and expected, we only track the set of symbols seen and expected for permutation phrases. In the standard algorithm, the number of states for processing a permutation phrase of length $n$ is computed as the sum of all $k$-permutations of $n$, $0<k\leq n$. In the modified algorithm  we only need to compute all $k$-combinations of $n$, $0<k\leq n$.

\smallskip\noindent\textbf{Related work.}
To the best of our knowledge, ours it the first approach to efficient LR parsing of permutation phrases. There are few works that extend top-down parsing methods for this purpose: In \cite{cameron1993:permphrases}, a modification to the LL parser is presented that keeps the $O(n)$ running time. In \cite{baars2004:functional} a way how to extend a parser combinator library is proposed. An XML parser presented in \cite{zhang2008:xmlparsing} uses a two-stack pushdown automaton to parse XML documents against an LL(1) grammar with permutation phrases.  

Algorithms for minimizing deterministic finite automata (DFA) \cite{hopcroft1971:minimizing,brzozowski1963:minimizing,moore1956:minimizing} could be used to reduce the states of LR(0) automaton. However they cannot be applied directly since minimizing an LR(0) automaton is different from minimizing a general DFA -- the content of the states must be taken into account to ensure proper shift/reduce actions. In addition, an extra step in the generation of the parsing table would be needed..

\section{Permutation Phrases in CFG grammars}
We assume the reader is familiar with the concepts of context-free grammars, finite automata, the family of LR parsers, and the underlying LR(0)/LR(1) automata. Without loss of generality, we assume that the CFG under consideration does not contain unreachable or non-generating nonterminals.

The right-hand side of a CFG rule is a sequence of grammar symbols called \emph{a simple phrase}. We refer to the set of simple phrases over an alphabet $\Sigma$ as $\Pi_{s}(\Sigma)$, i.e., $\Pi_{s}(\Sigma) = \Sigma^{*}$.  A \emph{permutation phrase} over a set of simple phrases $\{ \sigma_{1}, \sigma_{2},\ldots , \sigma_{n}\}$ is denoted by
$\llangle\, \sigma_{1} \,\|\, \sigma_{2} \,\|\, \ldots \,\|\, \sigma_{n} \,\rrangle$.
We consider a permutation phrase to be only a specific notation for a set of elements, with its semantics explicitly defined by the $expand$ function. 
Let  $\pi$ be the permutation phrase over the set $\{ A, B, C\}$, then $expand(\pi)$ returns all permutation options:
$$\pi = \llangle\, A \,\|\, B \,\|\, C \,\rrangle \mbox{ then } expand(\pi) = \{ ABC, ACB, BAC, BCA, CAB, CBA \}.$$
We denote the set of all permutation phrases over simple phrases from $\Pi_{s}(\Sigma)$ by $\Pi_{p}(\Sigma)$, it holds $\Pi_{p}(\Sigma) = 2^{\Sigma}$. 
We use the following notations for permutation phrases in the rest of this paper:  $|\pi|$ (size), $\pi \cup \pi'$ (union), $\pi \setminus \pi'$ (subtraction), $\sigma \in \pi$ (membership). A partition of $\pi$ into two subsets is denoted by $\{ \pi_{1}, \pi_{2} \}$.
The permutation phrases containing the same elements are considered equal.

Permutation phrases can be integrated into the right-hand sides of CFG rules as a shorthand notation for unordered content. Then each rule with a permutation phrase replaces $n!$ enumerated rules. Let $r$ be of the form:
$r: X \to Y \llangle\, A \,\|\, B \,\|\, CD \,\rrangle$ then
we refer to the set of equivalent enumerated rules as $expand(r)$:
\smallskip\par\noindent
\[
\begin{array}{lcl}
expand(r) & = \{ & X \to YAB CD, X \to Y A CD B, X \to Y B A CD, \\
& & X \to Y B CD A, X \to Y CD A B, X \to Y CD B A\,\, \} 
\end{array}
\]
\begin{definition}
Let $\Sigma$ be an alphabet. The set $\Pi(\Sigma)$ of grammatical phrases with intergated permutation phrases over $\Sigma$ and the expand function capturing their semantics
are defined as follows:
\begin{enumerate}
\setlength\itemsep{0.4em}
\item $\sigma \in \Pi_{s}(\Sigma)$ then $\sigma\in \Pi(\Sigma)$, $expand(\sigma) = \{ \sigma \}$,
\item $\sigma_{1}, \ldots, \sigma_{n} \in \Pi_{s}(\Sigma) \setminus\{\varepsilon\}$ then $\llangle \sigma_{1} \,\|\, \sigma_{2} \,\|\, \ldots \,\|\, \sigma_{n} \rrangle \in \Pi(\Sigma)$,\\
$expand(\llangle\, \sigma_{1} \,\|\, \sigma_{2} \,\|\, \ldots \,\|\, \sigma_{n} \,\rrangle) = \{ 
\mbox{ all permutations of } \sigma_{1}\sigma_{2}\ldots \sigma_{n}  \}$,
\item $\omega_{1}, \omega_{2}\in \Pi_{s}(\Sigma)\cup \Pi_{p}(\Sigma) $ then $\omega_{1} \omega_{2}\in \Pi(\Sigma)$,\\
$expand(\omega_{1}\omega_{2}) = expand(\omega_{1})\,\, expand(\omega_{2})$\footnote{The concatenation of sets $expand(\omega_{1})$ and $expand(\omega_{2})$.}.
\end{enumerate}
\end{definition}
Now we can define CFG with permutation phrases and its expanded grammar:
\begin{definition}
A \emph{context-free grammar with permutation phrases (CFGP)} is a 4-tuple $G = (N,T,P,S)$ where $N$ is a set of nonterminals, $T$ is a set of terminals, $S$ is the initial nonterminal, and $P$ is a finite set of rules
$P\subseteq N\times \Pi((N\cup T)^{*})$.
The \emph{expanded grammar} of $G$ is a CFG $G_{e} = (N, T, P_{e}, S)$ such that
$P_{e} = \bigcup_{r \in P} expand(r)$. 
\end{definition}
Note that each CFG $G$ is also a CFGP and in that case $G_{e} = G$. We refer to the rules of CFGP that contain permutation phrases as \emph{permutation rules} and we use the shorthand notation $\Pi(G)$ for $\Pi((N\cup T)^{*})$. In the rest of this paper, we consider the following grammars (unless stated otherwise): 
$G = (N,T,P,S)$ is a CFGP and $G_{e} = (N,T,P_{e},S)$ is the expanded CFG for $G$. In what follows, we consider permutation phrases over a set of symbols only.  Possible approaches to overcome these limitations are discussed in Section \ref{sec:conclusion}.

\section{LR(0) Items of Permutation Rules}
The LR(0) items (shortly \emph{items}) are rules with the dot in their right-hand side. The dot divides the right-hand side into two parts, possibly empty.
We first define the set of \emph{item phrases}, i.e., the possible right-hand sides of the items.
\begin{definition} 
Given an alphabet $\Sigma$ and $\omega\in\Pi(\Sigma)$, the set $IP(\omega)$ of item phrases of $\omega$ is defined as follows:
\begin{enumerate}
\item $\omega \in \Pi_{s}(\Sigma)$ is a simple phrase then $IP(\omega) = \{ \sigma_{1}\cdot^{(1)}\sigma_{2} \,|\, \sigma_{1}\sigma_{2} = \omega\}$,
\item $\omega\in \Pi_{p}(\Sigma)$ is a permutation phrase then 
\[
\begin{array}{lllll}
IP(\omega) & = & \{ \cdot^{(1)} \pi, \pi \cdot^{(1)} \}\,\, 
& \cup & \{ \pi_{1} \cdot^{(2)} \pi_{2} \,|\, \{ \pi_{1},\pi_{2} \} \mbox{ is a partition of } \pi \}.
\end{array}
\]
\item $\omega = \omega_{1}\omega_{2}$, $\omega_{1},\omega_{2}\in\Pi_{s}(\Sigma)\cup\Pi_{p}(\Sigma)$ is a concatenation of phrases then
$$IP(\omega) = \{ \omega_{1}'\omega_{2} \,|\, \omega_{1}'\in IP(\omega_{1})\} \cup \{ \omega_{1}\omega_{2}' \,|\, \omega_{2}'\in IP(\omega_{2})\} .$$
\end{enumerate}
\end{definition}
\begin{definition}
Let $r\in P$ be of the form: $X \to \omega$. The set of LR(0) items of $r$, denoted by $I(r)$, is defined by
$$I(r) = \{ [X\to \omega'] \,|\, \omega'\in IP(\omega) \}.$$
\end{definition}
We extended the definition with the items of permutation rules. An item of the form $\{ \pi_{1} \cdot^{(2)} \pi_{2} \}$, where $\{ \pi_{1}, \pi_{2}\}$ is a partition of some $\pi$ indicates that the elements of $\pi_{1}$ have already been seen in some order, while the elements in $\pi_{2}$ are expected to be seen in some order. It means that the items for permutation phrases do not care about the exact order of the elements.
 The dot is marked by $(2)$ meaning that the content of a permutation phrase is begin processed. We use dot without superscript to indicate that both $(1)$ and $(2)$ can be used.  We define the set of all items for a given CFGP as the union of items of its rules:
\begin{definition} The set of items for grammar $G$ is defined by
$I(G) = \bigcup_{r\in P} I(r)$.
\end{definition}

\section{Modified Algorithm for Generating LR(0) Automaton}
We first define the $next$ function which, for a given phrase $\omega$ of $G$, returns the set of grammar symbols that appear at the beginning of its expansions.
\begin{definition}
The function $next: \Pi(G) \to 2^{(N\cup T)}$ is defined by
\begin{itemize}
\item $\omega = \varepsilon$ then $next(\omega) = \emptyset$
\item $\omega = Y\omega'$, where $Y\in (N\cup T)$ then $next(\omega) = \{ Y \}$
\item $\omega = \llangle \sigma_{1} \,\|\, \sigma_{2} \,\|\, \ldots \,\|\,  \sigma_{n} \rrangle\, \omega'$ then $next(\omega) = \{ next(\sigma_{1}), next(\sigma_{2}), \ldots , next(\sigma_{n})\}$ 
\end{itemize}
Let $i\in I(G)$ be an item of the form $[X\to \alpha\cdot\beta]$, then $next(i) = next(\beta).$
\end{definition}
For a given item $i$ and a grammar symbol $Y$, the partial function $step$ returns the item that results from the movement of the dot within the item $i$ on symbol $Y$. The function is defined only if such movement is possible, i.e., $Y\in next(i)$.
\begin{definition}
Let $i\in I(G)$ such that $Y\in next(i)$. Then the partial function $step: I(G)\times (N
\cup T) \to I(G)$ is defined as follows (processing of the content of a permutation phrase is marked by a box):
\begin{enumerate}
 \setlength\itemsep{0.4em}
\item Matching 1st level (concatenation of phrases):
\begin{flushleft}
$i = [X\to\omega_{1} \cdot^{\scriptscriptstyle (1)}  \theta \omega_{2}]$, $\theta\in N\cup T\cup \Pi_{p}(G)$ then
 \end{flushleft}
\begin{enumerate}
 \setlength\itemsep{0.4em}
\item if $\theta = Y$  then
$step (i, Y) = [X\to\omega_{1} Y \cdot^{\scriptscriptstyle (1)}  \omega_{2}]$
\item if $\theta\in \Pi_{p}(G)$ then
$step(i, Y) = [X\to\omega_{1} \mbox{ \fbox{$\llangle\,Y\,\rrangle \cdot^{\scriptscriptstyle (2)} (\theta \setminus \{Y\})$} } \omega_{2}]$
\end{enumerate}
\item Matching 2nd level (the content of a permutation phrase):
\begin{flushleft}
$i = [X\to\omega_{1}  \,$\fbox{$\pi_{1} \cdot^{\scriptscriptstyle (2)} \pi_{2}$} $\omega_{2}]$, $\pi_{1},\pi_{2}\in \Pi_{p}(G)$ then
\end{flushleft}
\begin{enumerate}
 \setlength\itemsep{0.4em}
\item if $\pi_{2} = \llangle Y \rrangle$ then 
$step(i, Y) = [X\to\omega_{1} (\pi_{1}\cup\{ Y \}) \cdot^{\scriptscriptstyle (1)} \omega_{2}]$
\item if $|\pi_{2}|\geq 2$ then
$step(i, Y) = [X\to\omega_{1} \mbox{ \fbox{$(\pi_{1}\cup\{ Y \}) \cdot^{\scriptscriptstyle (2)} (\pi_{2} \setminus \{Y\})$} } \omega_{2}]$
\end{enumerate}
\end{enumerate}
\end{definition}

\begin{example}
Let $i_{0}: [X\to \cdot^{\scriptscriptstyle (1)} A \llangle \, B \,\|\, C \, \rrangle\, D]$. Using the notation $i\rightarrow^{step\,Y} j$ for $step(i,Y)=j$ we get: 
\begin{equation*}
\begin{array}{lll}
i_{0} & \rightarrow^{step\,A}  [ X\to A \cdot^{\scriptscriptstyle (1)}  \llangle \, B  \,\|\, C  \, \rrangle\,  D ] & \rightarrow^{step\,C}  [ X\to A\,\,   \llangle \, C \, \rrangle \cdot^{\scriptscriptstyle (2)} \llangle\, B  \rrangle D ] \rightarrow^{step\,B} \\
& \rightarrow^{step\,B}  [X \to A\,\,   \llangle \, B  \,\|\, C\, \rrangle \cdot^{\scriptscriptstyle (1)} D  ] & \rightarrow^{step\,D} 
 [X\to A\,\,   \llangle \, B  \,\|\, C \, \rrangle\,  D \cdot^{\scriptscriptstyle (1)} ]
\end{array}
\end{equation*}
\end{example}
The step function preserves the rule being recognized, meaning that both the original item and the resulting item belong to the set of items $I(r)$ for the same rule $r$:
\begin{proposition}
\label{lem:preserve_items}
Let $r\in P$ and $i\in I(r)$ then $(step(i,Y) = j) \Rightarrow j\in I(r).$
\end{proposition}
Actually, the superscripts of the dot help to preserve the rule being processed.
\begin{example} If  we do not mark the dot to determine the level being processed, then the item
$i = [X\to\llangle\, A\, \| \, B \,\rrangle \cdot  \llangle\, C\, \| \, D \,\rrangle]$
can originate from two different rules $r_{1}, r_{2}$ and $step(i, D)$ can result in two different options $j_{1}, j_{2}$:
\[
\begin{array}{lll}
r_{1}: X\to\llangle\, A\, \| \, B \,\rrangle \llangle\, C\, \| \, D \,\rrangle, \,\,\,\, & r_{2}:  X\to\llangle\, A\, \| \, B \,\|\, C\, \| \, D \,\rrangle\\
j_{1}: [X\to\llangle\, A\, \| \, B \,\rrangle \llangle \,C\, \rrangle \cdot  \llangle\, D \,\rrangle], \,\,\,\,& j_{2}: [X\to\llangle\, A\, \| \, B \| \, C \,\rrangle \cdot  \llangle\, D \,\rrangle]
\end{array}
\]
Marking the dot in $i$ by $(1)$ indicates that the top level is being processed - that corresponds to the rule $r_{1}$ and the resulting item $j_{1}$. Marking the dot by $(2)$ indicates that the permutation phrase is being processed - that corresponds to the rule $r_{2}$ and the resulting item $j_{2}$. 
\end{example}

\noindent We modify the standard algorithm for generating LR(0) automaton \cite{aho1986:compilers} as follows:
\medskip

\footnotesize
\noindent{\tt
\textbf{SetOfItems} PERM\_CLOSURE(I)
\{\\
\hspace*{15pt}    $J  = I$;\\
\hspace*{15pt}    \textbf{repeat} \\
\hspace*{30pt}    \textbf{for} (each item $[A\to\alpha \cdot^{\scriptscriptstyle (i)} \beta]$ in $J$) \\
\hspace*{45pt}    \textbf{for} (each nonterminal $B\in next(\beta)$)\\
\hspace*{60pt}    \textbf{for} (each production $B\to \gamma$ of $G$)\\   
\hspace*{75pt}    \textbf{if} ($[B\to\cdot^{\scriptscriptstyle (1)} \gamma]$ not in $J$)\\   
\hspace*{90pt}    add $[B\to\cdot^{\scriptscriptstyle (1)} \gamma]$ to $J$;	\\
\hspace*{15pt}    \textbf{until} no more items are added to $J$ on one round; \\
\hspace*{15pt}	   \textbf{return} $J$;\\
\}
}
\medskip\par\noindent
{\tt
\textbf{SetOfItems} PERM\_GOTO(I,Y)
\{\\
\hspace*{15pt}    $J  = \emptyset$;\\
\hspace*{15pt}     \textbf{for} (each item $[X\to \alpha \cdot \beta]$ of $I$ such that $Y\in next(\beta)$)\\
\hspace*{30pt}    add $step([X\to \alpha \cdot \beta], Y)$ to $J$;\\
\hspace*{15pt}    return $CLOSURE(J)$;\\
\}
}
\normalsize

\medskip\noindent The algorithm uses the generic functions $next$ and $step$ defined above, which can be applied to any phrase, including those with integrated permutation phrases. The algorithm for constructing LR(0) parsing table with shift/reduce actions remains unchanged. 
We define LR(0) automaton for CFGP $G$ as follows:
\begin{definition}
Let $G' = (N',T,P', S')$ be the augmented grammar of $G$ 
such that $N' = N\cup \{S'\}$ and $P'  = P \cup \{ S'\to S \}$.  
Then the \emph{LR(0) automaton for $G$} is a DFA $A = ( \Sigma, Q, \delta, I_{0}, F)$ such that
\begin{itemize}
\item $\Sigma = N\cup T$, $I_{0} = \{ [S' \to \cdot^{(1)} S] \}$,
\item $Q$ and $\delta$ are constructed by $PERM\_GOTO$ function, they are extended by an error state $\emptyset$ and transitions from/to this error state in a standard way to make $A$ complete.
\end{itemize}
\end{definition}
In addition to the grammars $G, G_{e}$ defined before, we consider the following definitions in the rest of this paper (unless stated otherwise):
$A= (\Sigma,  Q, \delta, I_{0}, F)$ is the complete LR(0) automaton for $G$,
$A_{e}= (\Sigma, Q_{e}, \delta_{e}, I_{0e}, F_{e})$ is the complete LR(0) automaton for $G_{e}$ constructed by the standard algorithm. 

\section{Correctness }
\label{sec:correctness}

In this section we define an 1:N mapping between the states of $A$ and $A_{e}$. We prove that $A$ reaches a state $I$ after reading an input $w$ if and only if $A_{e}$ reaches one of the states in $map(I,w)$ after reading the same input $w$. We also show that both automata report the same shift/reduce action(s) on the given input word.
\begin{example}
\label{ex:mapping}
Let $r: X\to \llangle A\,\|\,B \,\|\, C \,\rrangle\, D\in P$ be a rule of $G$ being processed by $A$. After the symbols $A, B, C$ have been matched in some order, $A$ reaches a state containing the item $i: X\to [\llangle A\,\|\,B \,\|\, C \,\rrangle\cdot D]$.  We refer to $\llangle A\,\|\,B \,\|\, C \,\rrangle$ as $\pi$, then $expand(\pi)$ is the set of all permutations of $\{ A,B,C\}$. Two situations may occur:

(1) $A$ contains exactly one such state $I$.  Then all expansions of $\pi$ are input paths for this state (see Fig. \ref{fig:mapping} a)). The $map$ function applied on $I$ returns $3!$ states of $A_{e}$ such that each of them handles the processing of a single input path for $I$. In this case the item $i$ is independent.

(2) $A$ contains more such states, assume there are two of them $I_{1}$ and $I_{2}$ as in  Fig. \ref{fig:mapping} b). Now a part of the expansions of $\pi$ are input paths for $I_{1}$ - namely a single input path $ABC$ - and the rest of the expansions of $\pi$ are input paths for $I_{2}$. 
In this case $i$ interferes with the item $i': X\to ABC\cdot D$. Since the item $i'$ has only  a single input path $ABC$, the item $i$ must accommodate to this limitation that results in generation of two different states $I_{1}$ and $I_{2}$, both containing $i$.  
\begin{figure}[t]
\includegraphics[width=0.8\textwidth]{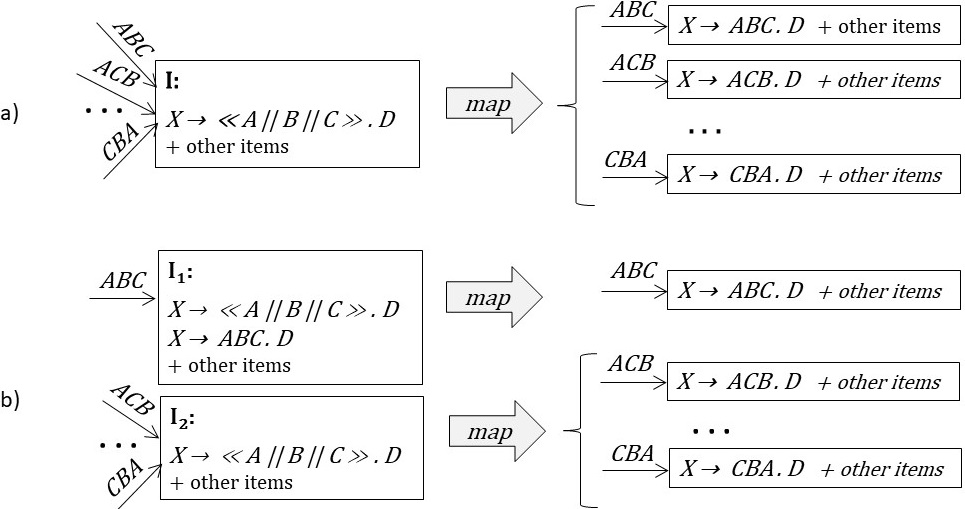}
\caption{The mapping of a) an independent rule b) a dependent rule.}
\label{fig:mapping}
\end{figure}
\end{example}

The $PERM\_GOTO$ function automatically handles rule interferences and generates as many states for $A$ as are actually needed. The more interferences occur, the more states of $A$ are generated, which results in a lower state reduction between $A$ and $A_{e}$. We define the function $paths$ to return input paths for a given state $I$ and the concept of independent CFG rules. It is sufficient to consider input paths no longer than a constant $max$, which is the maximum length of the phrase that appear before the dot in items of $I$. Thus, $paths(I)$ can be computed by examining up to $max$ computation steps backwards from the state $I$.
\begin{definition}
Let $I\in Q$ and $max = max\{|\alpha|\,|\, [X\to\alpha\cdot\beta]\in I\}$. Then the function $paths: Q\to 2^{\Sigma^{*}}$ returns \emph{input paths for} $I$ and is defined by:
$$paths(I) = \{ w \,|\, 0< |w| \leq max\mbox{ and there exists } I'\in Q \mbox{ such that } (I',w)\vdash_{A}^{*} (I,\varepsilon)\}.$$
\end{definition}
Note that $paths(I)$ is never empty and if $w$ is in $paths(I)$ then also all suffixes of $w$ (except empty word) are in $paths(I)$ . Given a set of strings $S$, we will refer to the set of suffixes of the elements of $S$ of length $n$ by $S /_{n}$.  Based on the definition of $\delta$ function, the following proposition holds that captures the relation between $paths(I)$ and items in $I$.  
\begin{proposition}
Let $I\in Q$ and $i\in I$ be an item of the form $I: [X\to\alpha\cdot\beta]$. Then
$paths(I)/_{|\alpha|} \subseteq expand(\alpha)$.
\end{proposition}
Rule independence means that a stronger condition holds and the sets $paths(I)/_{|\alpha|}$ and $expand(\alpha)$ equal for each of the rule items that appear in some $I\in Q$.
\begin{definition}
Let $I\in Q$ and $i\in I$ be an item of the form $I: [X\to\alpha\cdot\beta]$. Then the item $i$ is \emph{independent} in $I$ if and only if $paths(I)/_{|\alpha|} = expand(\alpha)$.
\end{definition}
\begin{definition}
Let $r\in P$, $r$ is \emph{independent} if and only if for each $i\in I(r)$: $i\in I$ implies $i$ is independent in $I$.
\end{definition}
We first define a function $map$ with two arguments:  for an input state $I\in Q$ and an input string $w$, it returns a state $I_{e}\in Q_{e}$, which handles the processing of the input paths of $paths(I)$ that are suffixes of $w$. The value of $map(I, w)$ is undefined if no suffix of $w$ is in $paths(I)$. We use auxilliary $map$ partial functions that map phrases and items of $G$ to phrases and items of $G_{e}$, respectively, with respect to an input $w$. 
\begin{definition} 
The partial function $map: \Pi(G)\times\Sigma^{*}\to\Pi(G_{e})$ is defined by:
\begin{equation}
map(\alpha, w) = \sigma \mbox{ if and only if } \alpha'\in expand(\alpha) \mbox{ and } \alpha' \mbox{ is a suffix of } w.
\end{equation}
The partial function $map: I(G)\times\Sigma^{*} \to 2^{I(G_{e})}$ is defined by:
$$map([X\to \alpha\cdot\beta], w) = \{ [X\to \alpha'\cdot \beta']\in I(G_{e}) \,|\, \alpha' = map(\alpha, w), \beta'\in expand(\beta)\}.$$
The function $map: Q\times\Sigma^{*} \to Q_{e}$ is defined by:
\begin{itemize}
\item $I = \emptyset$ then $map(I,w) = \emptyset$ for any $w\in\Sigma^{*}$ (error state),
\item $I\neq \emptyset$ and $w\neq w_{1}w_{2}$ such that $w_{2}\in paths(I)$  then $map(I,w) = \emptyset$,
\item $I\neq \emptyset$ and $w=w_{1}w_{2}$ such that $w_{2}\in paths(I)$  then
$map(I, w) = \bigcup_{i\in I} map(i, w).$
\end{itemize}
\end{definition}
The $map$ function without the string argument maps a state $I$ of $A$ to the set of the $A_{e}$ states that handle the processing of the paths in $paths(I)$ as shown in Example \ref{ex:mapping}. Note that only the phrases before the dot induce multiple states in $A_{e}$. 
\begin{definition}
The function $map: Q\to 2^{Q_{e}}$ is defined by:
$map(I) = \bigcup\limits_{\substack{\text{w$\,\in$}\\ \text{paths(I)}}} map(I,w).$
\end{definition}
Now we prove the key statements to show that $A$ is correct: the states reached by $A$ and $A_{e}$ on the same input can be related by the $map$ function and additionally, $A$ and $A_{e}$ return the same parser action. Note that both $A$ and $A_{e}$ perform $m$ computation steps to process an input of length $m$ since they are deterministic and complete.
\begin{lemma} Let $w\in \Sigma^{*}$. Let
\label{lem:correct1}
$(I_{0}, w) \vdash_{A}^{|w|} (J,\varepsilon), \mbox{ and } (I_{0e}, w) \vdash_{A_{e}}^{|w|} (J_{e},\varepsilon).$
Then $J_{e} = map(J,w)$.
\end{lemma}
\begin{proof}
We give a proof by induction on $|w|$. The base case $|w|=0$ trivially holds. Let us assume the statement holds for $n = k$ and let us have the following computations on the input $|w| = k+1$, where $w=w'Y$, $Y\in\Sigma$:
\begin{equation}
\label{eq:correct1_A_comp_kplus1}
(I_{0},w' Y) \vdash_{A}^{k} (I,Y) \vdash_{A} (J,\varepsilon) \mbox{ and } (I_{0e},w' Y) \vdash_{A_{e}}^{k} (I_{e},Y) \vdash_{A_{e}} (J_{e},\varepsilon).
\end{equation}
Then it also holds: $(I_{0},w') \vdash_{A}^{k} (I,\varepsilon)$  and  $(I_{0},w') \vdash_{A}^{k} (I,\varepsilon)$ and based on the induction hypothesis  $I_{e} = map(I,w).$
We need to prove  $J_{e} = map(J, w Y)$. It is sufficient to prove that mapping holds for kernel items\footnote{Kernel items are those that do not have the dot at the beginning of the right-hand side (rhs). The items with a dot at the beginning of rhs are non-kernel items.}  - $kernel(J_{e}) = map(kernel(J), w Y)$ - as that implies that the mapping holds for nonkernels as well and thus $J_{e} = map(J, wY).$ We define the subsets $I_{Y}\subseteq I$, $I_{eY} \subseteq I_{e}$ that participate in the computation step of $A$ and $A_{e}$, respectively, on the symbol $Y$:
\begin{figure}[t]
\includegraphics[width=0.9\textwidth]{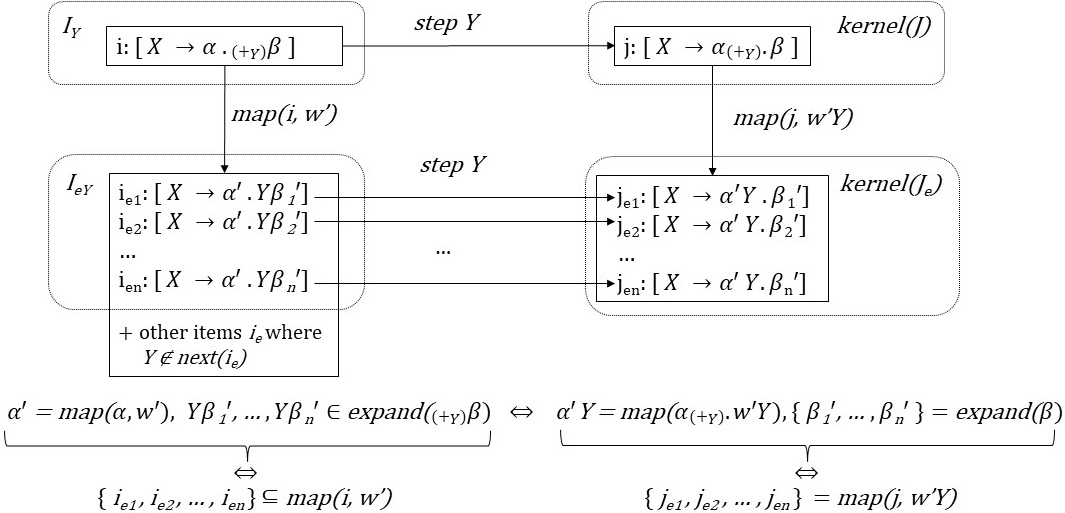}
\caption{The items participating in the last computation step of $A$ and $A_{e}$ (on the symbol $Y$) and their mutual relationships}
\label{fig:correct1proof}
\end{figure}
\begin{equation}
I_{Y} = \{ i\in I, Y\in next(i)\},\,\,\, I_{eY} =  \{ i_{e}\in I_{e}, Y\in next(i_{e})\}
\end{equation}
Based on the definition of $\delta$ and $\delta_{e}$ functions it holds
\begin{equation}
\delta(I,Y) = \delta(I_{Y}, Y) =J,\,\,\,\delta_{e}(I_{e},Y) =  \delta_{e}(I_{eY}, Y) = J_{e}
\end{equation}
\normalsize
If $I=\emptyset$ or $I_{Y}=\emptyset$, then $J=J_{e} = \emptyset$ and the statement clearly holds. Assume $I, I_{Y}\neq \emptyset$. The situation is depicted in Fig. \ref{fig:correct1proof}. 
We use the following shorthand notations:
\begin{itemize}
\item $\alpha_{(+Y)}$ to refer to the phrase resulting from adding $Y$ to the end of $\alpha$ - either by concatenation or by addition to a permutation phrase.
\item $_{(-Y)}\beta$, $Y\in next(\beta)$ to refer to the phrase resulting from removing $Y$ from the beginning of $\beta$ - either by trimming or by removal from a permutation phrase.
\end{itemize}
It is easy to see that
\begin{equation*}
\begin{array}{lllllll}
\multicolumn{3}{l}{j_{e}\in map(kernel(J), w'Y) \Rightarrow}\\[0.1cm]  
& \Rightarrow & \exists j\in kernel(J): j_{e}\in map(j, w' Y)   & \Rightarrow   \exists i\in I_{Y}: step(i, Y) = j  \Rightarrow  \\
 & \Rightarrow & \exists i_{e}\in I_{eY}: i_{e}\in map(i, w'), step(i_{e},Y) = j_{e}  & \Rightarrow  j_{e}\in kernel(J_{e}),\\[0.2cm]  
\multicolumn{3}{l}{j_{e}\in kernel(J_{e}) \Rightarrow} \\[0.1cm]  
& \Rightarrow & \exists i_{e}\in I_{eY}: step(i_{e}, Y) = j_{e}   & \Rightarrow \exists i\in I_{Y}: map(i, w') = i_{e}  \Rightarrow  \\
& \Rightarrow & \exists j\in J: j = step(i,Y), map(j, w'Y) = j_{e} & \Rightarrow   j_{e}\in map(kernel(J), w'Y). \qed
\end{array}
\end{equation*}
\end{proof}
\begin{theorem}
Let $w\in \Sigma^{*}$ and $Y\in\Sigma$. Let
$(I_{0}, w) \vdash_{A}^{*} (J,\varepsilon), \mbox{ and } (I_{0e}, w) \vdash_{A_{e}}^{*} (J_{e},\varepsilon).$
Then $J_{e}\in map(J)$.
\end{theorem}
The proof is based on Lemma \ref{lem:correct1}.  
\begin{theorem}
Let $w\in \Sigma^{*}$ and $Y\in\Sigma$. Let
$(I_{0}, w) \vdash_{A}^{*} (J,\varepsilon), \mbox{ and } (I_{0e}, w) \vdash_{A_{e}}^{*} (J_{e},\varepsilon).$
Then $$action(J, Y) = action(J_{e}, Y) \mbox{ or } J = J_{e} = \emptyset.$$
\end{theorem}
\begin{proof}
Based on Lemma \ref{lem:correct1} we get $J_{e} = map(J, w)$. Then it holds $J = \emptyset \Leftrightarrow J_{e} = \emptyset$. Let assume $J, J_{e} \neq \emptyset$. Based on the definition of $map$ function, an item $i\in J$ has a transition on $Y$ if and only if at least one of the items $i_{e}\in map(i,w)$ has a transition on $Y$. That implies $(shift)\in action(J, Y)$ if and only if $(shift)\in action(J_{e}, Y)$. At the same time, an item $i\in J$ is of the form $[X\to \alpha\cdot^{(1)}]$ if and only if $map(i,w) =\{i_{e}\}$ and $i_{e}$ is of the form $[X\to \sigma\cdot]$ where $\sigma = map(\alpha,w)$. This means $(reduce\,X\to\alpha)\in action(J, Y)$ if and only if $(reduce\,X\to\sigma)\in action(J_{e}, Y)$.
\end{proof}

\section{State Complexity}
\label{sec:complexity}
In this section, we analyze the difference between the number of states needed to process a permutation phrase in $A$ and to process all corresponding permutation options in $A_{e}$. We also discuss the differences in processing simple phrases in permutation rules. The greatest state reduction is achieved when a rule is independent, meaning that processing permutation phrases on the rule's right-hand side does not interfere with other rules or other parts of the same rule. 
\begin{definition}
Let $r\in P$, then the \emph{set of states processing the rule $r: X\to\omega$ in $A$ starting from $I_{0}$} is defined by
$\mathscr{I}(r, I_{0}) = \bigcup_{k=0}^{|\omega|} \mathscr{I}_{k}(r,I_{0}), \mbox{ where}$
\begin{enumerate}
\item $\mathscr{I}_{0}(r,I_{0}) = \{ I_{0} \}$, $I_{0}$ contains $[X\to \cdot\, \omega]$,
\item $\mathscr{I}_{k}(r,I_{0}) = \{ I \,|\,  I = \delta(I',Y)$,  where
\begin{itemize}
\item $I'\in \mathscr{I}_{k-1}(r,I_{0})$ and there exists $i\in I'\cap I(r)$ of the form $[X\to \omega_{1}\cdot\omega_{2}]$,\\ where $|\omega_{1}| = k-1$ and $Y\in next(\omega_{2})\}.$
\end{itemize}
\end{enumerate}
\end{definition}
Note that the set $\mathscr{I}_{k}(r,I_{0})$ contains all states reached from $I_{0}$ by processing the first $k$ symbols of $\omega$ and the set $\bigcup_{k=i}^{j} \mathscr{I}_{k}(r,I_{0})$ contains all states needed to process the subphrase of $\omega$ between the $i$-th and the $j$-th symbol.
If we replace $P, A, \delta$ with $P_{e}, A_{e}, \delta_{e}$, respectively, we get similar definition for the extended LR(0) automaton $A_{e}$. 
\begin{theorem}
Let $r\in P$ be an independent rule: $X\to\sigma_{0} \pi_{1}\sigma_{1} \ldots \sigma_{m-1}\pi_{m}\sigma_{m}$, 
 where each $\sigma_{i}\in\Pi_{s}(G)$ is a simple phrase, each $\pi_{i}\in\Pi_{p}(G)$ is a permutation phrase. Then the processing of $\sigma_{i}$/$\pi_{i}$ in $A$/$A_{e}$ requires the following number of states:
\begin{center}
\begin{tabular}{llllllll}
$A/\pi_{i}:$ & at most  $2^{|\pi_{i}|}$ states $\,\,(\in O(2^{|\pi_{i}|})$, & $\,\,A/\sigma_{i}:$ & at most  $|\sigma_{i}|$ states,\\
$A_{e}/\pi_{i}:$ & at least  $M \sum_{k=1}^{|\pi_{i}|} P(|\pi_{i}|,k)$ states $\,\,(\in \Omega (|\pi_{i}|!)$, & $\,\,A_{e}/\sigma_{i}$ & at least $M |\sigma_{i}|$  states.
\end{tabular}
\end{center}
 where the multiplication factor $M = \prod_{j=1}^{i-1}|\pi_{j}|!$ and $P(|\pi|,k) = \frac{|\pi_{i}|!}{(|\pi_{i}|-k)!}$ is the number of all $k$-permutations of $\pi_{i}$.
\end{theorem}
\begin{proof}
Processing of $\pi_{i}$ in $A$ starts in a state $I_{0}\in Q$ that contains the item of the form $[X\to \omega_{1} \cdot^{(1)} \pi_{i} \omega_{2}]$. Then $A$ passes the states that contain the following items:
\begin{align}
\label{eq:complex_A_items}
[X & \to \omega_{1} \pi_{1}\cdot^{(2)}\pi_{2}\omega_{2}], \{\pi_{1}, \pi_{2} \} \mbox{ is a partition of } \pi_{i} 
\mbox{ and } [X & \to \omega_{1}\pi_{i}\cdot^{(1)}\omega_{2}],
\end{align}
where $\pi_{1}$ can be any of the $k$-combinations of $\pi_{i}$ for $0 <k < |\pi_{i}|$. When we count all options for (\ref{eq:complex_A_items}), we get the number of the states needed for processing $\pi_{i}$ in $r$ starting from $I_{0}$:
\begin{equation}
\left| \,\, \bigcup_{k=|\omega_{1}| +1}^{|\omega_{1}|+|\pi_{i}|} \mathscr{I}_{k}(r,I_{0})\,\,\right| = \sum_{k=1}^{|\pi_{i}|}C(|\pi_{i}|,k) = 2^{|\pi_{i}|} - 1.
\end{equation}
Let $A_{e}$ be in some of the states of $I_{0e}\in map(I_{0})$, based on the definition of $map$ function, it must be
of the form $$[X\to w_{1} \cdot \sigma w_{2}], \mbox{ where } w_{1}\in expand(\omega_{1}), \sigma\in expand(\pi_{i})\mbox{ and } w_{2}\in expand(\omega_{2}).$$
While processing $\sigma$, $A_{e}$ passes states that contain items of the form 
\begin{equation}
\label{eq:complex_Ae_items}
[X\to w_{1} \sigma_{1}\cdot\sigma_{2}w_{2}], w_{1}\in expand(\omega_{1}), \sigma = \sigma_{1}\sigma_{2}, \sigma_{1}\neq\varepsilon, w_{2}\in expand(\omega_{2}),
\end{equation}
where 
$\sigma_{1}$ can be any of the $k$-permutations of $\pi_{1}$ for $0 <k \leq |\pi_{i}|$. 
When we count all the options for $\sigma_{1}\cdot\sigma_{2}$, we get the number of states needed for processing all expansions of $\pi_{i}$ starting from the state $I_{0e}$:
\begin{equation}
\label{eq:complex_perm}
\left| \,\, \bigcup_{k=|w_{1}|+1}^{|w_{1}|+|\sigma|} \mathscr{I}_{k}(r,I_{0e})\,\,\right|  =
\sum_{k=1}^{|\pi_{i}|}P(|\pi_{i}|, k) = \sum_{k=1}^{|\pi|} \frac{|\pi|!}{(|\pi_{i}|-k)!} \geq |\pi_{i}|! -1
\end{equation}
In order to count all options for (\ref{eq:complex_Ae_items}), we have to multiply (\ref{eq:complex_perm}) by the factor $M$ that equals the number of different $w_{1}\in expand(\omega_{1})$. That corresponds to the product of the numbers of all permutation options for permutation phrases that appear in $\omega_{1}$: $M = \sum_{j=1}^{i-1} |\pi_{j}|!$. The statements for a simple phrase $\sigma_{i}$ can be proved similarly. In this case the processing proceeds in the same way both in both $A$ and $A_{e}$; however, the multiplication factor must again be applied for $A_{e}$. \qed

\end{proof}
We analyzed the state reduction for independent rules. With dependent rules, some states of $A$ are split and in the worst case, the number of the states of $A$ equals to the number of the states of $A_{e}$. It cannot be lower, as a state of $A$ can be split into at most as many states as is the number of its input paths and that is exactly the number of $A_{e}$ states.  However, the real-world grammars typically do not contain extensive rule interferences.

\subsection{JSON Example}
We provide an example of JSON schema in the form of CFGP grammar and demonstrate the state reduction achieved by our modified algorithms. Consider the following CFGP grammar $G$ that define the content of the complex objects and arrays:
\[
\footnotesize
\begin{array}{lcl}
catalog & \to & 1: catalogItem \,|\, 2: catalog\,\, catalogItem\\
catalogItem & \to & 3: \llangle\, \mathtt{id} \,\|\, \mathtt{name} \,\|\, addresses  \,\rrangle\\
addresses & \to & 4: addressesItem \,\, addresses \,|\, 5: \varepsilon \\
addressesItem & \to & 6: \llangle\,  \mathtt{addressId} \,\|\, \mathtt{homeAddress} \,\|\, \mathtt{street} \,\|\, \mathtt{number} \,\|\, \mathtt{city} \,\|\, \mathtt{code}  \,\rrangle
\end{array}
\normalsize
\]
The right-hand sides of the rules $3$ and $6$ consist of a permutation phrases of length $3$ and $6$, respectively. It is easy to see that both rules are independent. When we construct the LR(0) automaton $A$ for $G$ and $A_{e}$ for the expanded grammar of $G$, we obtain the following number of states needed for processing the permutation rules\footnote{We also included the item having the dot at the beginning.} - note that the state reduction increases rapidly as the length of the permutation phrase grows:
\begin{center}
\begin{tabular}{lllllll}
$A$&/  rule 3: at most $2^3$ $= 8$, & $A$ & /  rule 6: at most $2^6 = 64$\\
$A_{e}$&/  rule 3: at least $\sum_{k=0}^{3} \frac{3!}{(3-k)!} = 16$, & $A_{e}$ & /  rule 6: at least $\sum_{k=0}^{6} \frac{6!}{(6-k)!} = 1975$
\end{tabular}
\end{center}

\section{Construction of SLR / canonical LR / LALR parsing tables}
We describe the modification to the standard algorithms for constructing SLR / canonical LR / LALR parsing tables \cite{aho1986:compilers} so that they can process CFGPs. Two functions are needed - $FIRST$ and $FOLLOW$ and we extend them to handle permutation rules:

\begin{flushleft}
Extension of the FIRST function:\\
$\omega = \llangle\,\sigma_{1} \,\|\, \sigma_{2} \,\|\, \ldots \,\|\, \sigma_{n}\,\rrangle$  then  $FIRST(\omega) = \bigcup_{0\leq i\leq n} FIRST(\sigma_{i})$,\\[0.1cm]
$\omega = \omega_{1}\omega_{2}, \omega_{1}, \omega_{2}\in \Pi_{s}(G)\cup\Pi_{p}(G)$ then 
\begin{itemize}
\item if $\varepsilon\in FIRST(\omega_{1})$ then $FIRST(\omega) = FIRST(\omega_{1})\cup FIRST(\omega_{2})$,
\item if $\varepsilon\notin FIRST(\omega_{1})$ then $FIRST(\omega) = FIRST(\omega_{1})$.
\end{itemize}
Extension of the FOLLOW function: Let  $r: X\to \omega_{1}\pi \omega_{2}$ be a rule of $G$. If $Y\in\pi$  then 
\begin{itemize}
\item for each $Y'\in \pi$, $Y'\neq Y$,  add $FIRST(Y')\setminus \{\varepsilon\}$ to $FOLLOW(Y)$ ,
\item add $FIRST(\omega_{2})\setminus \{\varepsilon\}$ to $FOLLOW(Y)$,
\item if  $\omega_{2} = \varepsilon$ or $\varepsilon\in FIRST(\omega_{2})$ then add $FOLLOW(X)$ to $FOLLOW(Y)$.
\end{itemize}

\end{flushleft}
We get LR(1) items by adding lookahead to the LR(0) items. The body of the repeat loop in the closure function for LR(1) items is modified as follows\footnote{We use the notation $_{(-B)}\beta$ as in the proof of Lemma \ref{lem:correct1}.}:

\begin{flushleft}
\footnotesize
\noindent{\tt
\hspace*{15pt}    \textbf{for} (each item $[A\to\alpha \cdot^{\scriptscriptstyle (i)} \beta, a]$ in $J$) \\
\hspace*{30pt}    \textbf{for} (each nonterminal $B\in next(\beta)$)\\
\hspace*{45pt}    \textbf{for} (each production $B\to \gamma$ of $G$)\\   
\hspace*{60pt}    \underline{\textbf{for} (each symbol $b$ in  $FIRST(_{(-B)}\beta a)$}  \\
\hspace*{75pt}    \textbf{if} ($[B\to\cdot^{\scriptscriptstyle (1)} \gamma]$ not in $J$)\\   
\hspace*{90pt}    add $[B\to\cdot^{\scriptscriptstyle (1)} \gamma, b]$ to $J$;	\\
}
\end{flushleft}
\normalsize
Assume $A$ and $A_{e}$ are LR(1) automata for $G$ and $G_{e}$, respectively, constructed using the $PERM\_CLOSURE$ function above. 
The $map$ function for an LR(1) item and an input string $w$ maps the LR(0) part of the item in the same way as for LR(0) items and does not manipulate the lookahead part. Let $I$ be a state of $A$. Each of the mapped states $I_{e}$ contains all expansions of the phrases that appear after the dots in $I$. When constructing parsing table for canonical LR parser, the states of $A$ are split based on the lookahead only if the mapped states of $A_{e}$ are also split, preserving the state reduction rate.  Similarly, when merging states for an LALR parser, the state reduction rate remain unaffected.

\section{Conclusion and Future Work}
\label{sec:conclusion}
We presented a modification of LR parsing algorithms that, in practical cases, generates significantly smaller parsing tables. For independent rules, the number of states needed for processing a permutation phrase $\pi$ of size $n$ in LR(0) automaton is reduced from $\Omega(n!)$ to $O(2^{n})$. The reduction in the number of states increases with the size of $\pi$ as well as its placement within the right-hand side of a rule. The more permutation phrases appear before $\pi$, the higher the reduction. In addition to providing a more efficient approach for processing permutation phrases in existing languages, we hope that the findings of this work will also assist language designers in making informed decisions about incorporating permutation phrases into their specifications.

Our algorithm does not support nested simple phrases and optional elements within a permutation phrase. It is pretty straightforward to include both of them. For nested phrases, another level for the dot needs to be introduced and the $step$ function must be extended to handle that level. It is required that, within a permutation phrase, one nested simple phrase is not a prefix of another to avoid conflicts. Optional elements require the modification of the $next$ function and they cannot conflict with the set of symbols that can follow given permutation phrase. In both cases the limitations could be avoided by parallel processing of more items.  It would be beneficial to extend the algorithm to handle nested simple phrases and optional elements without limitations.  Another possible direction for future work is to explore in detail the relationship between the number and type of rule interferences and the resulting reduction in states.

%

\end{document}